\documentclass{article}

\usepackage{amsmath,amsxtra,amssymb,amsthm,amsfonts}
\usepackage{bbm}
\usepackage{mathtools}
\usepackage{color,graphicx}
\usepackage[margin=1.05in]{geometry}
\usepackage[dvipsnames]{xcolor}
\usepackage{etoolbox}
\usepackage{url}
\usepackage[section]{placeins}
\usepackage{booktabs}

\usepackage[linktocpage=true,colorlinks,citecolor=blue,linkcolor=blue,pagebackref]{hyperref}
\usepackage{cleveref}
\usepackage{caption}
\usepackage{subcaption}

\usepackage{pgfplots}
\usepackage{pgf}
\usepackage{tikz}
\usetikzlibrary{patterns}
\usetikzlibrary{arrows.meta}
\usepgfplotslibrary{patchplots} % LATEX and plain TEX
\usetikzlibrary{pgfplots.patchplots} % LATEX and plain TEX
\pgfplotsset{width=9cm,compat=1.5.1}

\usepackage{enumitem}
\usepackage{tikz-cd}
\usetikzlibrary{positioning}
\usetikzlibrary{calc}

\newcommand{\bij}{%
  \hookrightarrow\mathrel{\mspace{-15mu}}\rightarrow
}

\makeatletter
\newcommand{\xbij}[2][]{%
  \lhook\joinrel
  \ext@arrow 0359\rightarrowfill@ {#1}{#2}%
  \mathrel{\mspace{-15mu}}\rightarrow
}
\makeatother

\usepackage{algorithm,algpseudocode}
\algnewcommand\algorithmicforeach{\textbf{for each}}
\algdef{S}[FOR]{ForEach}[1]{\algorithmicforeach\ #1\ \algorithmicdo}

\algnewcommand{\algorithmicand}{\textbf{ and }}
\algnewcommand{\algorithmicor}{\textbf{ or }}
\algnewcommand{\algorithmicnot}{\textbf{ not }}
\algnewcommand{\algorithmicbreak}{\textbf{break}}
\algnewcommand{\algorithmiccontinue}{\textbf{continue}}
\algnewcommand{\AND}{\algorithmicand}
\algnewcommand{\OR}{\algorithmicor}
\algnewcommand{\NOT}{\algorithmicnot}
\algnewcommand{\BREAK}{\algorithmicbreak}
\algnewcommand{\CONTINUE}{\algorithmiccontinue}

\algrenewcommand\algorithmiccomment[1]{\hfill{\color{gray}\(\triangleright\)
#1}}

\makeatletter
\newcommand{\LeftComment}[1]{%
  \Statex
  {%
    \setlength\leftskip{\ALG@thistlm}%
    \noindent
    \color{gray}\(\triangleright\)\enspace#1\par%
  }%
}

\newcommand{\LeftCommentIndent}[1]{%
  \Statex
  {%
    \setlength\leftskip{\dimexpr\ALG@thistlm+\algorithmicindent}%
    \noindent
    \color{gray}\(\triangleright\)\enspace#1\par%
  }%
}
\makeatother

\def\Lin{\mathcal{L}}

\def\0{\mathbf{0}}
\def\1{\mathbf{1}}

\numberwithin{equation}{section}
\theoremstyle{plain}

\newtheorem{theorem}{Theorem}[section]
\newtheorem{proposition}[theorem]{Proposition}
\newtheorem{lemma}[theorem]{Lemma}
\newtheorem{corollary}[theorem]{Corollary}

\newtheorem{definition}{Definition}
\newtheorem{fact}{Fact}

\numberwithin{equation}{section}
\theoremstyle{plain}

\title{
  A polynomial-time algorithm for recognizing high-bandwidth graphs
}

\author{
  Luis M. B. Varona\textsuperscript{1,2,3}
}

\begin{document}

\maketitle

\begin{abstract}
  An unweighted, undirected graph $G$ on $n$ nodes is said to have \emph{bandwidth} at most $k$ if its nodes can be labelled from $0$ to $n - 1$ such that no two adjacent nodes have labels that differ by more than $k$. It is known that one can decide whether the bandwidth of $G$ is at most $k$ in $O(n^k)$ time and $O(n^k)$ space using dynamic programming techniques. For small $k$ close to $0$, this approach is effectively polynomial, but as $k$ scales with $n$, it becomes superexponential, requiring up to $O(n^{n - 1})$ time (where $n - 1$ is the maximum possible bandwidth). In this paper, we reformulate the problem in terms of bipartite matching for sufficiently large $k \ge \lfloor (n - 1)/2 \rfloor$, allowing us to use Hall's marriage theorem to develop an algorithm that runs in $O(n^{n - k + 1})$ time and $O(n)$ auxiliary space (beyond storage of the input graph). This yields polynomial complexity for large $k$ close to $n - 1$, demonstrating that the bandwidth recognition problem is solvable in polynomial time whenever either $k$ or $n - k$ remains small.\\\medskip

  \medskip

\end{abstract}

\addtocounter{footnote}{1}
\footnotetext{Department of Politics \& International Relations,
Mount Allison University, Sackville, NB, Canada E4L 1E4}
\addtocounter{footnote}{1}
\footnotetext{Department of Economics, Mount Allison University,
Sackville, NB, Canada E4L 1E4}
\addtocounter{footnote}{1}
\footnotetext{Department of Mathematics \& Computer Science, Mount
Allison University, Sackville, NB, Canada E4L 1E4}

\section{Introduction}\label{section:1}

Let $G$ be an unweighted, undirected finite graph on $n$ nodes with no self-loops. The \emph{bandwidth} of $G$ under some labelling (or ``layout'') $\pi$ of its nodes is the minimum non-negative integer $k \in \{0, 1, \ldots, n - 1\}$ such that $\lvert \pi(u) - \pi(v) \rvert \le k$ for every edge $\{u, v\}$ in $G$; we denote this quantity by $\beta_\pi(G)$. A long-standing problem in graph theory and combinatorics is to find a layout $\pi$ for which $\beta_{\pi}(G)$ is small. Many computational processes involve sparse symmetric matrices that can be viewed as adjacency matrices, and reducing the bandwidth of the corresponding graphs helps improve properties like cache locality and low storage space. As such, graph bandwidth reduction has many practical applications in engineering and scientific computing, including solving linear systems, approximating partial differential equations, optimizing circuit layout, and more \cite{MS14}.

There are two variants of the bandwidth reduction problem: \emph{minimization} (finding a layout $\pi$ that minimizes or nearly minimizes $\beta_{\pi}(G)$) and \emph{recognition} (finding a layout $\pi$ such that $\beta_{\pi}(G) \le k$ given some fixed $k \ge 0$). Although most bandwidth reduction applications rely on minimization algorithms (see, e.g., \cite{MS14}), bandwidth recognition (the primary subject of this paper) also has applications in other fields of theoretical mathematics and computer science. In computational complexity theory, for instance, \cite{RS83} showed that a graph on $n$ nodes with bandwidth $\le k$ requires at most $\min\{2k^2 + k + 1, 2k\log_2{n}\}$ pebbles in the computational pebble game. Quantum information theorists, too, are interested in determining whether the density matrix of a quantum state (or rather, the associated graph) has bandwidth $\le k$, since this is necessary (although not sufficient) for a property called ``$k$-incoherence'' \cite{JMP25}.

In 1978, \cite{GGJK78} presented the first bandwidth recognition algorithm, designed to determine whether a graph with $m$ edges and $n$ nodes has bandwidth at most $2$ in $O(m + n)$ time. This was soon followed by the first general-case algorithm to determine whether a graph has bandwidth at most $k$ for arbitrary $k$, requiring $O(n^{k + 1})$ time and space \cite{Sax80}. This algorithm was later improved upon by \cite{GS84}, who reduced both time and space complexity to $O(n^k)$. The most recent relevant work was conducted by \cite{DCM99} and \cite{CSG05}---although the main contributions of both studies were (exact) bandwidth minimization algorithms, these algorithms iteratively invoke their own recognition algorithms as subroutines. All existing approaches, however, are optimized for low-bandwidth recognition, exhibiting polynomial complexity only for small $k$---when $k$ scales with $n$, these algorithms exceed even exponential complexity, requiring up to $O(n^{n - 1})$ time (since $n - 1$ is the maximum possible bandwidth of a graph on $n$ nodes).

This paper addresses this gap in the literature, being the first to tackle the problem of high-bandwidth recognition (for $k$ close to $n - 1$). We present an algorithm that requires $O(n^{n - k + 1})$ time and $O(n)$ auxiliary space, making it effectively polynomial when $n - k$ is small. Our algorithm begins by iterating over all possible assignments of the first $n - k - 1$ nodes. We then exploit Hall's marriage theorem in each iteration to efficiently verify whether the remaining nodes can be placed to achieve bandwidth $\le k$, breaking early from the loop if a valid layout is found. This represents the first work to explicitly address the bandwidth recognition problem in the general case since 1984 \cite{GS84}, and the first to advance general-case bandwidth recognition algorithms in any capacity (indirect or otherwise) since 2005 \cite{CSG05}.

\subsection{Organization of the paper}\label{subsection:1.1}

We begin in Section \ref{section:2} by introducing the notation used throughout the paper and covering preliminary definitions and results. In Section \ref{section:3}, we build on these concepts, using some implications of Hall's marriage theorem to lay the theoretical groundwork for the algorithm. We then provide a detailed description of our algorithm in Section \ref{section:4}, accompanied by a formal complexity analysis and comparisons to other bandwidth recognition algorithms. In Section \ref{section:5}, we benchmark our algorithm against these previous approaches, demonstrating that it is indeed generally more efficient for large $k$ close to the maximum possible bandwidth of $n - 1$. Finally, we close in Section \ref{section:6} with some open questions and suggestions for future work.

\section{Preliminaries and definitions}\label{section:2}

We exclusively consider unweighted, undirected, finite graphs with no self-loops. Given a graph $G$, we write $V(G)$ for the node set of $G$, $E(G)$ for the edge set of $G$, and $A(G)$ for the adjacency matrix of $G$. Although less frequently, we also use $\varepsilon_G(v)$ to denote the eccentricity of $v$ in $G$ (i.e., the maximum shortest-path distance between $v$ and any other node in $G$) and $N_{G,k}(v)$ to denote the $k$-hop neighbourhood of $v$ in $G$ (i.e., the set of nodes in $G$ at distance between $1$ and $k$ from $v$).

We follow the standard convention of $0$-based indexing for all data structures and algorithms. Given an array \texttt{arr} and a hash map \texttt{map}, we write $\texttt{arr}[i]$ for the element of \texttt{arr} at index $i$ and $\texttt{map}[k]$ for the value in \texttt{map} associated with the key $k$. Moreover, we assume constant-time arithmetic and array access in all our complexity evaluations (rather than measure complexity in terms of bit operations).

Finally, on the topic of function notation: we use $\operatorname{im}(f)$ to denote the image of the map $f$, $g : A \bij B$ to indicate that $g : A \to B$ is bijective, and $h : C \hookrightarrow D$ to indicate that $h : C \to D$ is injective.

We now introduce further terminology---as well as formalize some existing terms already introduced in Section \ref{section:1}---necessary to understand our algorithm.

\begin{definition}[Layout]\label{definition:1}
    Let $G$ be a graph on $n$ nodes. A \emph{layout} $\pi$ of $G$ is a bijective map $\pi : V(G) \bij \{0, 1, \ldots, n - 1\}$, labelling the nodes of $G$ with distinct consecutive non-negative integers starting from $0$.
\end{definition}

\begin{definition}[Bandwidth]\label{definition:2}
    Let $G$ be a graph on $n$ nodes and $\pi$ be a layout of $G$. The \emph{layout bandwidth} of $G$ with respect to $\pi$, denoted by $\beta_{\pi}(G)$, is the minimum non-negative integer $k \ge 0$ such that $\lvert \pi(u) - \pi(v) \rvert \le k$ for all edges $\{u, v\} \in E(G)$. The \emph{graph bandwidth} (or simply \emph{bandwidth}) of $G$, denoted by $\beta(G)$, is the minimum value of $\beta_{\varphi}(G)$ over all layouts $\varphi$ of $G$.
\end{definition}

Before moving on, we note some basic facts regarding graph bandwidth that will later prove useful.

\begin{fact}\label{fact:1}
    Given a graph $G$ on $n$ nodes, $\beta(G)$ is always at most $n - 1$, since the maximum difference in labels of any two nodes $u, v \in V(G)$ is achieved when $\pi(u) = 0$ and $\pi(v) = n - 1$.
\end{fact}

\begin{fact}\label{fact:2}
    For any graph $G$, $\beta(G)$ is equal to the maximum bandwidth of its connected components, since connected components share no edges and can therefore be arranged sequentially in any layout.
\end{fact}

The focus of this paper is developing an efficient algorithm to determine whether $\beta(G) \le k$ for some graph $G$ and integer $k$. In particular, we are concerned with \emph{high-bandwidth} graph recognition---determining whether $\beta(G) \le k$ for ``large'' $k$. Specifically, we require $\left\lfloor \frac{n - 1}{2} \right\rfloor \le k < n - 1$, so that mappings which we call left and right partial layouts have nonempty domains that do not overlap with each other.

\begin{definition}[Partial layouts]\label{definition:3}
    Let $G$ be a graph on $n$ nodes and $k \in \left\{\left\lfloor \frac{n - 1}{2} \right\rfloor, \left\lfloor \frac{n - 1}{2} \right\rfloor + 1, \ldots, n - 2\right\}$ be an integer. Then:
    \begin{enumerate}[label=(\roman*)]
        \item A \emph{left partial layout} of $G$ (with respect to bandwidth $k$) is an injective map $\mathcal{L} : \{0, 1, \ldots, n - k - 2\} \hookrightarrow V(G)$, determining the first $n - k - 1$ nodes placed by some layout of $G$.
        \item A \emph{right partial layout} of $G$ (w.r.t. bandwidth $k$) is an injective map $\mathcal{R} : \{k + 1, k + 2, \ldots, n - 1\} \hookrightarrow V(G)$, determining the last $n - k - 1$ nodes placed by some layout of $G$.
        \item A left partial layout $\Lin$ and a right partial layout $\mathcal{R}$ (w.r.t. bandwidth $k$) are \emph{compatible} if there exists a layout $\pi$ such that $\pi^{-1}$ agrees with $\Lin$ and $\mathcal{R}$ on their domains (equivalently, if $\operatorname{im}(\Lin) \cap \operatorname{im}(\mathcal{R}) = \emptyset$). In such a case, $\pi$ \emph{induces} $\Lin$ and $\mathcal{R}$ as its left and right partial layouts (w.r.t. bandwidth $k$).
    \end{enumerate}
\end{definition}

The domain of any left or right partial layout (with respect to bandwidth $k$) consists of precisely $n - k - 1$ positions, so it is now clear why we restrict our attention to the case that $\left\lfloor \frac{n - 1}{2} \right\rfloor \le k < n - 1$. If $k = n - 1$, then $n - k - 1 = 0$, so any left or right partial layout (w.r.t. bandwidth $k$) will have an empty domain. (Indeed, this is a trivial case---as noted in Fact \ref{fact:1}, $\beta(G) \le n - 1$ for every graph $G$ on $n$ nodes, and any layout suffices.) If $k < \left\lfloor \frac{n - 1}{2} \right\rfloor$, then $2(n - k - 1) > n$, indicating that the domains of any left and any right partial layout (w.r.t. bandwidth $k$) will not be disjoint.

This would defeat the entire purpose of separately specifying which nodes occupy the leftmost and rightmost positions of a layout to reduce the search space for bandwidth recognition. As will soon be formalized in Lemma \ref{lemma:3.1}, whether a layout $\pi$ of $G$ satisfies $\beta_{\pi}(G) \le k$ depends solely on which nodes occupy these outer positions---namely, $\{0, 1, \ldots, n - k - 2\}$ on the left and $\{k + 1, k + 2, \ldots, n - 1\}$ on the right---since only edges between these two regions can have endpoints whose labels differ by more than $k$. Keeping the domains disjoint allows us to treat these two regions as independent, so that we may enumerate all possible left partial layouts and, for each, determine whether a valid right partial layout exists that induces bandwidth $\le k$. (Theorem \ref{theorem:3.2} provides an efficient criterion for this.)

\section{Preliminary results}\label{section:3}

We are now prepared to lay the mathematical foundations for our algorithm.

\begin{lemma}\label{lemma:3.1}
    Let $G$ be a graph on $n$ nodes and $k \in \left\{\left\lfloor \frac{n - 1}{2} \right\rfloor, \left\lfloor \frac{n - 1}{2} \right\rfloor + 1, \ldots, n - 2\right\}$ be an integer. A layout $\pi$ of $G$ satisfies $\beta_{\pi}(G) \le k$ if and only if $\{\pi^{-1}(i), \pi^{-1}(k + j + 1)\} \notin E(G)$ for all $0 \le i < n - k - 1$ and $i \le j < n - k - 1$.
\end{lemma}

\begin{proof}
    Let $G$, $k$, and $\pi$ be as defined above, and suppose that $\beta_{\pi}(G) \le k$. Then by definition, no edge has endpoints whose images under $\pi$ differ by more than $k$. For any fixed $i$ and $j$ with $0 \le i < n - k - 1$ and $i \le j < n - k - 1$, the positions $i$ and $k + j + 1$ differ by $(k + j + 1) - i \ge k + 1 > k$, so it follows that $\{\pi^{-1}(i), \pi^{-1}(k + j + 1)\} \notin E(G)$, and we have proven the forward direction of the lemma.

    We now prove the backward direction by contrapositive: suppose that $\beta_{\pi}(G) > k$. Then there exists an edge $\{u, v\} \in E(G)$ with $\lvert \pi(u) - \pi(v) \rvert > k$. Without loss of generality, assume that $\pi(u) < \pi(v)$, so $\pi(v) - \pi(u) \ge k + 1$. Observe that we can bound $\pi(u)$ from below and above by $\pi(u) \ge 0$ and
    \[\pi(u) \le \pi(v) - (k + 1) \le (n - 1) - (k + 1) < n - k - 1;\]
    similarly, we have $\pi(v) \ge \pi(u) + k + 1$ and $\pi(v) \le n - 1$. Finally, setting $i = \pi(u)$ and $j = \pi(v) - (k + 1)$ (so that $0 \le i < n - k - 1$ and $i \le j < n - k - 1$), we see that
    \[\{u, v\} = \{\pi^{-1}(\pi(u)), \pi^{-1}(\pi(v))\} = \{\pi^{-1}(i), \pi^{-1}(k + j + 1)\}\]
    forms an edge. By contrapositive, it follows that $\beta_{\pi}(G) \le k$ if no such edge exists, and we are done.
\end{proof}

It then follows that whether a layout $\pi$ of a graph $G$ on $n$ nodes satisfies $\beta_{\pi}(G) \le k$ depends solely on the left and right partial layouts it induces (at least when $k \ge \left\lfloor \frac{n - 1}{2} \right\rfloor$), motivating the following definition.

\begin{definition}[Feasible pair]\label{definition:4}
    Let $G$ be a graph on $n$ nodes, $k \in \left\{\left\lfloor \frac{n - 1}{2} \right\rfloor, \left\lfloor \frac{n - 1}{2} \right\rfloor + 1, \ldots, n - 2\right\}$ be an integer, and $\Lin$ and $\mathcal{R}$ be compatible left and right partial layouts of $G$ (with respect to bandwidth $k$). The pair $(\Lin, \mathcal{R})$ is called \emph{feasible} (w.r.t. bandwidth $k$) if $\beta_{\pi}(G) \le k$ for any layout $\pi$ that induces $\Lin$ and $\mathcal{R}$ as its left and right partial layouts.
\end{definition}

Using only Lemma \ref{lemma:3.1}, it is already possible to formulate a na\"ive algorithm to check whether $\beta(G) \le k$. Although we refrain from a comprehensive analysis for brevity, we briefly describe this algorithm below.
\begin{enumerate}
    \item Enumerate all left partial layouts of $G$. There are
    \[\frac{n!}{(n - (n - k - 1))!} \in O(n^{n - k - 1})\]
    of these, since we are taking all $(n - k - 1)$-permutations without repetition of all $n$ nodes.
    \item For each left partial layout $\Lin$, enumerate all right partial layouts of $G$ compatible with $\Lin$. There are
    \[\frac{(k + 1)!}{((k + 1) - (n - k - 1))!} \in O((k + 1)^{n - k - 1}) = O(n^{n - k - 1})\]
    such right partial layouts, since we are taking all $(n - k - 1)$-permutations without repetition of the remaining $n - (n - k - 1) = k + 1$ nodes.
    \item For each of the resulting $(\Lin, \mathcal{R})$ pairs, check whether $(\Lin, \mathcal{R})$ is feasible, and return early if so. By Lemma \ref{lemma:3.1}, we must verify $\{\Lin(i), \mathcal{R}(k + j + 1)\} \notin E(G)$ for all $0 \le i < n - k - 1$ and $i \le j < n - k - 1$ for each pair, amounting to
    \[\sum_{i = 0}^{n - k - 2} (n - k - i - 1) = \frac{(n - k - 1)(n - k)}{2} \in O((n - k)^2) = O(n^2)\]
    operations total.
\end{enumerate}
This procedure requires $O(n^{n - k - 1} \cdot n^{n - k - 1} \cdot n^2) = O(n^{2(n - k)})$ time, which is already effectively polynomial for large $k$ close to $n - 1$. Still, it is possible to derive an even more efficient approach---rather than enumerate all $O(n^{n - k - 1})$ right partial layouts compatible with each left partial layout $\Lin$, we can determine whether a valid $(\Lin, \mathcal{R})$ pair (given our fixed $\Lin$) exists far more quickly using Hall's marriage theorem.

\begin{theorem}\label{theorem:3.2}
    Let $G$ be a graph on $n$ nodes, $k \in \left\{\left\lfloor \frac{n - 1}{2} \right\rfloor, \left\lfloor \frac{n - 1}{2} \right\rfloor + 1, \ldots, n - 2\right\}$ be an integer, and $\Lin$ be a left partial layout of $G$ (with respect to bandwidth $k$). Additionally, define the sets
    \[A_j \coloneqq \{v \in V(G) \setminus \operatorname{im}(\Lin) : \{\Lin(i), v\} \notin E(G) \text{ for all } i \le j\}\]
    for all $0 \le j < n - k - 1$. Then there exists a right partial layout $\mathcal{R}$ (w.r.t. bandwidth $k$) compatible with $\Lin$ such that $(\Lin, \mathcal{R})$ is feasible if and only if
    \[\lvert A_j \rvert \ge n - k - j - 1 \quad \text{for all} \quad 0 \le j < n - k - 1.\]
\end{theorem}

\begin{proof}
    Let $G$, $k$, $\Lin$, and the $A_j$'s be as defined above. We now consider whether there exists a compatible right partial layout $\mathcal{R}$ such that $(\Lin, \mathcal{R})$ is feasible. By Lemma \ref{lemma:3.1} and Definition \ref{definition:4}, $(\Lin, \mathcal{R})$ is feasible if and only if $\{\Lin(i), \mathcal{R}(k + j + 1)\} \notin E(G)$ for all $0 \le i < n - k - 1$ and $i \le j < n - k - 1$. By construction of the $A_j$'s, this is equivalent to requiring $\mathcal{R}(k + j + 1) \in A_j$ for all $0 \le j < n - k - 1$.

    As such, we can reformulate this as a bipartite matching problem, where $\operatorname{domain}(\mathcal{R}) = \{k + 1, k + 2, \ldots, n - 1\}$ is on one side, $V(G) \setminus \operatorname{im}(\Lin)$ is on the other, and each $k + j + 1 \in \operatorname{domain}(\mathcal{R})$ is adjacent to all $v \in A_j \subseteq V(G) \setminus \operatorname{im}(\Lin)$. Hall's marriage theorem then tells us that we can construct such a matching if and only if for all subsets $S \subseteq \{0, 1, \ldots, n - k - 2\}$,
    \[\left\lvert \bigcup_{j \in S} A_j \right\rvert \ge \lvert S \rvert.\]
    It is useful to observe here that the $A_j$'s are nested, with $A_0 \supseteq A_1 \supseteq \cdots \supseteq A_{n - k - 2}$. Therefore, we have
    \[\bigcup_{j \in S} A_j = A_{\min S},\]
    and so we see that Hall's condition for a feasible $(\Lin, \mathcal{R})$ pair is equivalent to having $\lvert A_{\min S} \rvert \ge \lvert S \rvert$ for all subsets $S$.

    For any fixed $j$, the most restrictive constraint arises when $\min S = j$ and $\lvert S \rvert$ is maximized---namely, when $S = \{j, j + 1, \ldots, n - k - 2\}$. This gives $\lvert S \rvert = n - k - j - 1$, so indeed $\lvert A_{\min S} \rvert \ge \lvert S \rvert$ only if $\lvert A_j \rvert \ge n - k - j - 1$ for all $0 \le j < n - k - 1$. Conversely, if $\lvert A_j \rvert \ge n - k - j - 1$ for all $0 \le j < n - k - 1$, then for any subset $S$ with $\min S = j$, we have
    \[\lvert A_{\min S} \rvert = \lvert A_j \rvert \ge n - k - j - 1 = \lvert \{j, j + 1, \ldots, n - k - 2\} \rvert \ge \lvert S \rvert,\]
    so $\lvert A_{\min S} \rvert \ge \lvert S \rvert$ if $\lvert A_j \rvert \ge n - k - j - 1$ for all $0 \le j < n - k - 1$.

    We have therefore shown that there exists a right partial layout $\mathcal{R}$ such that $(\Lin, \mathcal{R})$ is feasible if and only if $\lvert A_j \rvert \ge n - k - j - 1$ for all $0 \le j < n - k - 1$, exactly as desired.
\end{proof}

In addition to providing a condition for when it is possible to find a feasible $(\Lin, \mathcal{R})$ pair for a given left partial layout $\Lin$, these results also allow us to explicitly construct such an $\mathcal{R}$.

\begin{corollary}\label{corollary:3.3}
    Let $G$, $k$, $\Lin$, and the $A_j$'s be as defined in Theorem \ref{theorem:3.2}, and suppose that the Hall conditions $\lvert A_j \rvert \ge n - k - j - 1$ are satisfied for all $0 \le j < n - k - 1$. Then we can construct a right partial layout $\mathcal{R}$ compatible with $\Lin$ such that $(\Lin, \mathcal{R})$ is feasible as follows: for each $j = n - k - 2, n - k - 3, \ldots, 0$ (in decreasing order), define $\mathcal{R}(k + j + 1)$ to be any node in $A_j$ not already selected in a previous step.
\end{corollary}

\begin{proof}
    It suffices to show that at each step $j$, there exists at least one element of $A_j$ not yet selected to be in the image of $\mathcal{R}$. When processing some fixed index $j$ (counting down from $n - k - 2$), we have already assigned exactly $n - k - j - 2$ nodes. By hypothesis, there exists a right partial layout $\mathcal{R}$ such that $(\Lin, \mathcal{R})$ is feasible; Theorem \ref{theorem:3.2} tells us that $\lvert A_j \rvert \ge n - k - j - 1 = (n - k - j - 2) + 1$ is a necessary condition for this, so at least one element of $A_j$ must still be available, and we are done.
\end{proof}

Finally, before presenting the algorithm itself in Section \ref{section:4}, we collect some lower bounds on graph bandwidth that serve as early termination conditions for bandwidth recognition.

\begin{proposition}\label{proposition:3.4}
    Let $G$ be a graph with $m$ edges and $n$ nodes, and define the quantities $\alpha(G)$ and $\gamma(G)$ by
    \[\alpha(G) \coloneqq \max\limits_{v \in V(G)} \max\limits_{1 \le k \le \varepsilon_G(v)} \left\lceil \frac{\lvert N_{G,k}(v) \rvert}{2k} \right\rceil \quad \text{and} \quad \gamma(G) \coloneqq \min\limits_{v \in V(G)} \max\limits_{1 \le k \le \varepsilon_G(v)} \left\lceil \frac{\lvert N_{G,k}(v) \rvert}{k} \right\rceil.\]
    Then:
    \begin{enumerate}[label=(\roman*)]
        \item $\alpha(G)$ and $\gamma(G)$ both bound $\beta(G)$ from below.
        \item $\alpha(G)$ and $\gamma(G)$ are both computable in $O(mn)$ time and $O(n)$ auxiliary space.
        \item Neither of these bounds consistently dominates the other.
    \end{enumerate}
\end{proposition}

\begin{proof}
    Propositions 4 and 6 of \cite{CSG05} give exactly the statements of (i) and (ii); note that they include the node itself in their definition of a $k$-hop neighbourhood---for which they also use slightly different notation---so our $\lvert N_{G,k}(v) \rvert$ corresponds to their $\lvert N_k(v) \rvert - 1$. The paper also gives constructive procedures to compute both $\alpha(G)$ and $\gamma(G)$ in $O(mn)$ time (instead of merely proving that it is possible to do so in theory), allowing us to actually invoke these bounds in our algorithm later in Section \ref{section:4}. Although space complexity claims are not explicitly presented therein, we also note that the dominating cost lies in computing the $k$-hop neighbourhood sizes via breadth-first search from each node, which requires only $O(n)$ auxiliary space beyond storage of the input graph.

    For (iii), \cite[p. 360]{CSG05} also note that there are graphs for which $\alpha$ is a tighter bound and others for which $\gamma$ is a tighter bound; we illustrate this ourselves with some specific examples. For instance, although we omit the actual calculations here for brevity, it is straightforward to see that the complete multipartite graph $K_{n,n}$ gives
    \[\alpha(K_{n,n}) = \left\lceil \frac{n}{2} \right\rceil < n = \gamma(K_{n,n})\]
    for all $n \ge 2$, whereas the lollipop graph $L_{n,n}$ gives
    \[\alpha(L_{n,n}) = \left\lceil \frac{n}{2} \right\rceil > 2 = \gamma(L_{n,n})\]
    for all $n \ge 5$. As such, it is useful to apply both bounds to create early termination conditions for any bandwidth recognition algorithm, including ours.
\end{proof}

We are now in possession of all the mathematical preliminaries needed to develop our algorithm.

\section{The algorithm}\label{section:4}

Once again, let $G$ be a graph on $n$ nodes and $k \in \left\{\left\lfloor \frac{n - 1}{2} \right\rfloor, \left\lfloor \frac{n - 1}{2} \right\rfloor + 1, \ldots, n - 2\right\}$ be an integer; we seek an answer to the question of whether $\beta(G) \le k$. As we have seen, Theorem \ref{theorem:3.2} allows us to check whether a given left partial layout $\Lin$ admits a compatible right partial layout $\mathcal{R}$ such that $(\Lin, \mathcal{R})$ is feasible---namely, by verifying that $\lvert A_j \rvert \ge n - k - j - 1$ for all $0 \le j < n - k - 1$ (where the $A_j$'s are as defined therein). However, computing each $\lvert A_j \rvert$ from scratch would require iterating over all $k + 1$ unplaced nodes in $V(G) \setminus \operatorname{im}(\Lin)$ and checking adjacency to $\Lin(i)$ for all $0 \le i \le j$, amounting to
\[\sum_{j = 0}^{n - k - 2} (k + 1)(j + 1) = (k + 1) \cdot \frac{(n - k - 1)(n - k)}{2} \in O(k(n - k)^2)\]
operations. We therefore introduce two additional data structures---a hash map \texttt{blockedMap} and an array \texttt{blockedArr}---to help us verify this condition more efficiently (namely, in $O(nk)$ time per left partial layout).

Fix a left partial layout $\Lin$ and, for convenience, denote the set of unplaced nodes by $V' \coloneqq V(G) \setminus \operatorname{im}(\Lin)$. For each unplaced node $v \in V'$, we define $\texttt{blockedMap}[v]$ to be the minimum index $i$ such that $\{\Lin(i), v\} \in E(G)$ (or $\infty$ by convention, if no $\Lin(i)$ is adjacent to $v$). Observe that $v \in A_j$ if and only if $\texttt{blockedMap}[v] > j$, since each $A_j \subseteq V'$ consists precisely of those nodes in $V'$ that are not adjacent to any $\Lin(i)$ with $i \le j$. We then define \texttt{blockedArr} to be an array containing the nodes in $V'$, sorted in increasing order by their corresponding \texttt{blockedMap} values. This allows us to compute each $\lvert A_j \rvert$ in $O(\log{k})$ time: since $\lvert V' \rvert = k + 1$ and a node $v \in V'$ belongs to $A_j$ if and only if $\texttt{blockedMap}[v] > j$, we have
\[\lvert A_j \rvert = (k + 1) - \bigl\lvert \{v \in V' : \texttt{blockedMap}[v] \le j\} \bigr\rvert.\]
The second term is precisely the number of nodes in \texttt{blockedArr} whose \texttt{blockedMap} values are at most $j$, which can be computed via binary search on the sorted array in $O(\log{k})$ time.

If taken in isolation, we could then say that since we consider $j \in \{0, 1, \ldots, n - k - 2\}$, it would take us $O((n - k)\log{k})$ time to compute all the $\lvert A_j \rvert$'s and thus determine the existence of a feasible $(\Lin, \mathcal{R})$ pair for our fixed $\Lin$. However, we also need to take into consideration the costs of building \texttt{blockedMap} and sorting \texttt{blockedArr}. Building \texttt{blockedMap} takes $O(k(n - k))$ time, since for each of the $k + 1$ nodes in $V'$, we check adjacency to each of the $n - k - 1$ nodes in $\operatorname{im}(\Lin)$. Meanwhile, sorting \texttt{blockedArr} takes $O(k\log{k})$ time, using any standard comparison-based algorithm.

Together with the cost of verifying all Hall conditions, the total time per left partial layout is then $O((n - k)\log{k} + k(n - k) + k\log{k})$. Depending on the value of $k$ relative to $n$, either the first or the second term can dominate; since they are all in $O(nk)$, we adopt this as our upper bound on time complexity per left partial layout. (We defer a more thorough verification of this bound to Subsection \ref{subsection:4.1}.)

Note that \texttt{blockedArr} also facilitates the construction of a right partial layout $\mathcal{R}$ such that $(\Lin, \mathcal{R})$ is feasible, as described in Corollary \ref{corollary:3.3}. Rather than explicitly compute all $A_j$'s (which would require unnecessary space allocation) and track which nodes have already been assigned, we can simply use the last $n - k - 1$ elements of \texttt{blockedArr} (i.e., those with the largest \texttt{blockedMap} values) to construct $\mathcal{R}$. We formalize this procedure below.

\begin{proposition}\label{proposition:4.1}
    Let $G$, $k$, $\Lin$, and the $A_j$'s be as defined in Theorem \ref{theorem:3.2}, and suppose that the Hall conditions $\lvert A_j \rvert \ge n - k - j - 1$ are satisfied for all $0 \le j < n - k - 1$. Then the right partial layout $\mathcal{R}$ defined by
    \[\mathcal{R}(k + j + 1) \coloneqq \textup{\texttt{blockedArr}}[2k - n + j + 2]\]
    for all $0 \le j < n - k - 1$ forms a feasible pair with $\Lin$.
\end{proposition}

\begin{proof}
    Corollary \ref{corollary:3.3} guarantees that when constructing $\mathcal{R}$ by iterating over $j = n - k - 2, n - k - 3, \ldots, 0$ (in decreasing order) and assigning any unused node from $A_j$ to $\mathcal{R}(k + j + 1)$ at each step, at least one such node is always available. We show that $\texttt{blockedArr}[2k - n + j + 2]$ is always such a node.

    As in the statement of the proposition, define the right partial layout $\mathcal{R}$ by
    \[\mathcal{R}(k + j + 1) \coloneqq \texttt{blockedArr}[2k - n + j + 2]\]
    for all $0 \le j < n - k - 1$, and recall that \texttt{blockedArr} is an array of the $k + 1$ nodes in $V' \coloneqq V(G) \setminus \operatorname{im}(\Lin)$, sorted in increasing order by their corresponding \texttt{blockedMap} values. The last $n - k - 1$ elements of this array occupy indices $2k - n + 2$ through to $k$, so these are precisely the elements in the image of $\mathcal{R}$.

    Theorem \ref{theorem:3.2} tells us that it suffices to verify that $\mathcal{R}(k + j + 1) \in A_j$ for all $0 \le j < n - k - 1$ to show that $(\Lin, \mathcal{R})$ is feasible. By construction, a node $v$ belongs to $A_j$ if and only if $\texttt{blockedMap}[v] > j$. Since
    \[\lvert A_j \rvert = (k + 1) - \bigl\lvert \{v \in V' : \texttt{blockedMap}[v] \le j\} \bigr\rvert,\]
    the Hall condition $\lvert A_j \rvert \ge n - k - j - 1$ (which holds by our initial hypothesis) implies that
    \[\bigl\lvert \{u \in V' : \texttt{blockedMap}[u] \le j\} \bigr\rvert \le (k + 1) - (n - k - j - 1) = 2k - n + j + 2.\]
    Therefore, nodes at indices $2k - n + j + 2$ and beyond in \texttt{blockedArr} have \texttt{blockedMap} values exceeding $j$, placing them in $A_j$. This directly implies that
    \[\mathcal{R}(k + j + 1) = \texttt{blockedArr}[2k - n + j + 2] \in A_j\]
    for all $0 \le j < n - k - 1$, so $(\Lin, \mathcal{R})$ is indeed feasible, as desired.
\end{proof}

Finally, we remark that although the algorithm as presented technically applies to disconnected graphs as well, it is more efficient in practice to run it separately on each connected component, since the bandwidth of any graph is equal to the maximum bandwidth of its connected components (Fact \ref{fact:2}). Having said that, the full procedure is given in the Algorithm \ref{algorithm:1} pseudocode.

\begin{algorithm}[!ht]
    \caption{Determine whether $\beta(G) \le k$ and, if so, find a layout $\pi$ of $G$ such that $\beta_{\pi}(G) \le k$}
    \label{algorithm:1}
    \begin{algorithmic}[1]
        \Function{HighBandwidthRecognition}{graph $G$, integer $k$} \Comment{For $\left\lfloor \frac{\lvert V(G) \rvert - 1}{2} \right\rfloor \le k < \lvert V(G) \rvert - 1$}
        \State let $n \gets \lvert V(G) \rvert$ \Comment{Simply for convenience of notation}

        \mbox{}

        \If{$k < \max\left\{\alpha(G), \gamma(G)\right\}$} \Comment{By Proposition \ref{proposition:3.4}(iii), it is helpful to compute both bounds}
        \State \Return $(\texttt{false},\, \texttt{null})$
        \EndIf

        \mbox{}

        \LeftComment{We use this array to implicitly represent a right partial layout (it is more efficient to preallocate it here outside of the loop, as its entries will be overwritten in each iteration anyway)}
        \State let $\texttt{right} \gets [\text{an empty array of length } n - k -1]$

        \mbox{}

        \ForEach{left partial layout $\Lin : \{0, 1, \ldots, n - k - 2\} \hookrightarrow V(G)$}
        \State let $V' \gets V(G) \setminus \operatorname{im}(\Lin)$
        \State let $\texttt{blockedMap} \gets [\text{an empty hash map from nodes to integers}]$

        \mbox{}

        \LeftComment{$v \in A_j$ if and only if $\texttt{blockedMap}[v] > j$, allowing for efficient computation of the $\lvert A_j \rvert's$}
        \ForEach{$v \in V'$}
	    \State let $\texttt{blockedMap}[v] \gets \min\bigl\{\{\infty\} \cup \{i : \{\Lin(i), v\} \in E(G)\}\bigr\}$
        \EndFor

        \mbox{}

        \State let $\texttt{blockedArr} \gets [\text{an array of all } v \in V' \text{ sorted by their } \texttt{blockedMap}[v] \text{ values}]$
        \State let $j \gets 0$
        \State let $\texttt{validRight} \gets \texttt{true}$

        \mbox{}

        \While{$j < n - k - 1$\AND $\texttt{validRight} = \texttt{true}$}
        \LeftCommentIndent{We can efficiently compute the $\bigl\lvert \{v \in V' : \texttt{blockedMap}[v] \le j\} \bigr\rvert$ term via binary search on the sorted \texttt{blockedArr} array, treating it as $\bigl\lvert \{i : \texttt{blockedMap}[\texttt{blockedArr}[i]] \le j\} \bigr\rvert$ instead}
        \State let $\lvert A_j \rvert \gets k + 1 - \bigl\lvert \{v \in V' : \texttt{blockedMap}[v] \le j\} \bigr\rvert$

        \mbox{}

        \LeftComment{By Theorem \ref{theorem:3.2}, a feasible $(\Lin, \mathcal{R})$ pair exists if and only if this inequality holds for all $j$}
        \If{$\lvert A_j \rvert \ge n - k - j - 1$}
        \LeftCommentIndent{By Corollary \ref{corollary:3.3}/Proposition \ref{proposition:4.1}, this is a valid way to construct a right partial layout that forms a feasible pair with $\Lin$}
        \State let $\texttt{right}[j] \gets \texttt{blockedArr}[2k - n + j + 2]$
        \State let $j \gets j + 1$
        \Else
        \State let $\texttt{validRight} \gets \texttt{false}$
        \EndIf
        \EndWhile

        \mbox{}

        \If{$\texttt{validRight} = \texttt{true}$}
        \State let $\mathcal{R} \gets [\text{the right partial layout mapping each $k + j + 1\mapsto \texttt{right}[j]$}]$
        \LeftComment{We can assign the nodes not in $\Lin$ or $\mathcal{R}$ to the remaining indices in any arbitrary order}
        \State \Return $(\texttt{true},\, [\text{any layout of } G \text{ inducing } \Lin \text{ and } \mathcal{R} \text{ as left and right partial layouts}])$
        \EndIf
        \EndFor

        \mbox{}

        \State \Return $(\texttt{false},\, \texttt{null})$
        \EndFunction
    \end{algorithmic}
\end{algorithm}

\subsection{Complexity analysis}\label{subsection:4.1}

We now analyze the time and space complexity of Algorithm \ref{algorithm:1}, proving that it requires $O(n^{n - k + 1})$ time and $O(n)$ auxiliary space to determine whether a graph $G$ with $m$ edges and $n$ nodes has bandwidth at most $k$.

The algorithm begins by computing the $\alpha$ and $\gamma$ lower bounds on $\beta(G)$, which, by Proposition \ref{proposition:3.4}(ii), each require $O(mn)$ time. If either bound exceeds $k$, the algorithm terminates immediately; otherwise, the main loop iterates over all left partial layouts (with respect to bandwidth $k$) of $G$. As mentioned earlier in Section \ref{section:3}, this is precisely the set of all $(n - k - 1)$-permutations without repetition of $V(G)$ (which has $n$ elements), amounting to
\[\frac{n!}{(n - (n - k - 1))!} \in O(n^{n - k - 1})\]
left partial layouts total. Standard techniques for enumerating partial permutations allow each successive partial layout to be generated in $O(n)$ amortized time---for instance, \cite{Jim98} gives an efficient procedure that first generates all $(n - k - 1)$-combinations of $n$ elements, then generates all permutations of each combination. Below, we show that the processing cost associated with each left partial layout is $O(nk)$, dominating this $O(n)$ (amortized) enumeration cost.

For each left partial layout $\Lin$, the algorithm performs the following operations. First, building the \texttt{blockedMap} hash map (lines $10$--$12$) requires checking, for each of the $k + 1$ unplaced nodes in $V'$, adjacency to each of the $n - k - 1$ nodes in $\operatorname{im}(\Lin)$; this requires $O((k + 1)(n - k - 1)) = O(k(n - k))$ time. Second, sorting the \texttt{blockedArr} array (line $13$) using any standard comparison-based algorithm requires $O(k\log{k})$ time. Third, verifying the Hall conditions from Theorem \ref{theorem:3.2} (lines $16$--$24$) requires computing $n - k - 1$ $\lvert A_j \rvert$'s via binary search on \texttt{blockedArr}, which has exactly $k + 1$ elements; this requires $O((n - k - 1)\log{(k + 1)}) = O((n - k)\log{k})$ time. (We also construct the \texttt{right} array---representing a right partial layout---concurrently with this step, avoiding a second pass over the data.)

Combining these contributions, we see that the time required to process each left partial layout is $O(k(n - k) + k\log{k} + (n - k)\log{k})$. Each of these terms is bounded above by $O(nk)$, with the first dominating when $k$ is closer to $\left\lfloor \frac{n - 1}{2} \right\rfloor$ (so $n - k > \log{k}$) and the second dominating when $k$ is closer to $n - 1$ (so $\log{k} > n - k$). Multiplying this $O(nk)$ cost by the $O(n^{n - k - 1})$ left partial layouts, the total time requirement for the main loop is then $O(n^{n - k - 1} \cdot nk) = O(n^{n - k} k) = O(n^{n - k + 1})$. (We simplify from $O(n^{n - k}k)$ to $O(n^{n - k + 1})$ because our $k \ge \left\lfloor \frac{n - 1}{2} \right\rfloor$ constraint ensures that $k \in \Theta(n)$, implying that $O(n^{n - k}k)$ is no tighter as a bound.) This always dominates the $O(mn)$ cost of computing the $\alpha$ and $\gamma$ lower bounds: the first nontrivial case is $k = n - 2$ (since all graphs on $n$ nodes have bandwidth $\le n - 1$), and for all $k \le n - 2$, we have $O(n^{n - k + 1}) \supseteq O(n^3) \supseteq O(mn)$ (since $m \in O(n^2)$).

If a feasible $(\Lin, \mathcal{R})$ pair is not found within this loop, then we simply return a negative answer, and the overall time complexity of the algorithm is thus $O(n^{n - k + 1})$ as seen above. Otherwise, if a feasible $(\Lin, \mathcal{R})$ pair is found, then constructing a full layout that induces $\Lin$ and $\mathcal{R}$ as its left and right partial layouts requires assigning the remaining $2k - n + 2$ nodes to indices $\{n - k, n - k + 1, \ldots, k\}$ in arbitrary order. This requires $O(k)$ time, which is again dominated by the preceding $O(n^{n - k + 1})$ time requirement. Therefore, we conclude that the algorithm's time complexity is bounded above by $O(n^{n - k + 1})$.

We now turn to space complexity. By Proposition \ref{proposition:3.4}(ii), computing the $\alpha$ and $\gamma$ lower bounds requires $O(n)$ auxiliary space. The main loop similarly uses only $O(n)$ auxiliary space: the \texttt{right} array contains $n - k - 1 \in O(n)$ elements, the \texttt{blockedArr} array contains $k + 1 \in O(n)$ elements, and the \texttt{blockedMap} hash map stores $k + 1 \in O(n)$ key-value pairs. The current left partial layout $\Lin$ in any given iteration can also be represented in $O(n)$ space as an array (similarly to how \texttt{right} is used to implicitly represent a right partial layout). Thus, the overall auxiliary space complexity (beyond storage of the input graph, which will typically require $O(m + n)$ space depending on the implementation) is $O(n)$, as required.

\subsection{Comparison to previous work}\label{subsection:4.2}

Whereas \cite{GS84}'s $O(n^k)$ algorithm is effectively polynomial for small $k$ close to the minimum bandwidth of $0$ but superexponential when $k$ scales with $n$, our $O(n^{n - k + 1})$ algorithm is instead polynomial for large $k$ close to the maximum possible bandwidth of $n - 1$. The papers \cite{DCM99} and \cite{CSG05}, on the other hand, present bandwidth minimization algorithms containing implicit recognition subroutines that are iteratively run for incrementing values of $k$ until an affirmative result is obtained. (This is never explicitly addressed by the authors, but for instance, lines $4$--$6$ of the pseudocode in \cite[Table 1]{DCM99} constitute a recognition algorithm in their own right.) No complexity claims are made about either the minimization algorithms presented or the implicit recognition subroutines, but any thorough analysis---which we omit here for brevity---shows that they are not polynomial for any regime of $k$ values. (They do, however, incorporate effective pruning techniques that, as demonstrated by our benchmarks in Section \ref{section:5}, consistently outperform \cite{GS84}'s algorithm.) Hence, our algorithm is the only known method that is polynomial in time complexity for large $k$.

One limitation of our approach is that we are restricted to checking bandwidth $\le k$ for $k \ge \left\lfloor \frac{n - 1}{2} \right\rfloor$, whereas other existing bandwidth recognition algorithms can check bandwidth $\le k$ for all $k \ge 0$. In any case, this is not a practical concern. Just as \cite{GS84}'s $O(n^k)$ algorithm is superexponential in complexity for large $k$ and thus impractical for such cases, we would never want to run our $O(n^{n - k + 1})$ algorithm for small $k < \left\lfloor \frac{n - 1}{2} \right\rfloor$ even if we theoretically could.

\section{Computational data}\label{section:5}

\subsection{Methodology}\label{subsection:5.1}

Using the interface provided by the Julia package MatrixBandwidth.jl (v0.3.0) \cite{Var25}, we implemented our algorithm and compared it against preexisting implementations of \cite{GS84}'s, \cite{DCM99}'s, and \cite{CSG05}'s bandwidth recognition algorithms. All benchmarks were executed using Julia 1.12.4 on a 2021 MacBook Pro running Fedora Linux Asahi Remix 42 with an Apple M1 Pro chip, 32 GB of RAM, and 10 cores (8 performance Firestorm and 2 efficiency Icestorm). The full implementation of our algorithm and the corresponding benchmark tests can be found in the companion code for this paper \cite{Var26}. (Note also that preexisting code in MatrixBandwidth.jl v0.3.0 computed the $\alpha$ and $\gamma$ lower bounds on $\beta$---which were used as early termination conditions for all four algorithms, not just ours---in $O(n^3)$ time instead of $O(mn)$ due to working with dense matrix representations of graphs. However, this was the only discrepancy with our work, and the overall algorithmic complexity remained the same.)

In particular, we compared our algorithm against $(1)$ the main recognition algorithm from \cite{GS84}, $(2)$ the recognition subroutine of the MB-ID minimization algorithm from \cite{DCM99}, and $(3)$ the recognition subroutine of the \texttt{LAYOUT\_LEFT\_TO\_RIGHT} minimization algorithm from \cite{CSG05}. This excluded the MB-PS algorithm from \cite{DCM99} and the \texttt{LAYOUT\_BOTH\_WAYS} algorithm from \cite{CSG05}. We refrained from testing MB-PS because it requires a ``depth'' parameter that is nontrivial to optimize, and performance is highly sensitive to this choice. We refrained from testing \texttt{LAYOUT\_BOTH\_WAYS} because it did not have a preexisting implementation in MatrixBandwidth.jl v0.3.0, being considerably more difficult to implement. Furthermore, results from the original papers showed that neither \texttt{LAYOUT\_LEFT\_TO\_RIGHT} nor \texttt{LAYOUT\_BOTH\_WAYS} consistently outperformed the other, and although MB-PS typically (but not always) outperformed MB-ID with an optimal depth parameter, running times nonetheless remained comparable.

We benchmarked all algorithms on both affirmative (bandwidth $\le k$) and negative (bandwidth $> k$) test cases. These instances were generated using MatrixBandwidth.jl's \texttt{random\_banded\_matrix} function, which---given an integer $n \ge 1$, an integer $\psi \in \{0, 1, \ldots, n - 1\}$, and a floating-point number $p \in (0, 1]$ as input---outputs the adjacency matrix of a graph $G$ on $n$ nodes with $\beta(G) \le \psi$. In particular, it assigns $n$ nodes the labels $\{0, 1, \ldots, n - 1\}$, then for each pair of nodes whose labels differ by at most $\psi$, it includes an edge with probability $p$. It also ensures that for all $0 < k \le \psi$, there exists at least one edge whose endpoints' labels differ by exactly $k$, even when $p$ is very small (although for the values of $p$ that we use---described below---this makes no difference in practice). This is the same method that \cite{DCM99} used to generate test cases for their experimental results, with the sole difference that $\texttt{random\_banded\_matrix}(n, \psi, p)$ guarantees at least one edge exists at each distance $k \in \{1, 2, \ldots, \psi\}$.

Given a fixed $n$ and $k$, we needed both affirmative test cases where $\beta(G) \le k$ and negative test cases where $\beta(G) > k$. It was straightforward to generate affirmative cases by randomly choosing $\psi \in \{k - 2, k - 1, k\}$ (not restricting it to $k$ so that variety was present, yielding more realistic test instances) and taking the graph $G = \texttt{random\_banded\_matrix}(n, \psi, p)$. (Here $p$ was randomly sampled from $[0.3, 0.6]$---the same range of values that \cite{DCM99} used.) We then applied random layouts $\pi$ to $G$ until $\beta_{\pi}(G) > k$, ensuring that the algorithms would need to discover a valid layout inducing bandwidth $\le k$ rather than have access to one from the start.

Negative cases were more difficult to generate. Taking $G = \texttt{random\_banded\_matrix}(n, \psi, p)$ guarantees that $\beta(G)$ is at most $\psi$, but it does not prevent $\beta(G)$ from being much less than $\psi$, especially when $p$ is small to moderate and $n$ is large. Therefore, we repeatedly sampled $\psi \in \{k + 1, k + 2, k + 3\}$ and $p \in [0.85, 0.95]$ and called $G = \texttt{random\_banded\_matrix}(n, \psi, p)$ until both $\max\{\alpha(G), \gamma(G)\} \le k$ and $\beta(G) > k$ were satisfied. (The first condition avoided early bounds-based termination, and the second condition was checked using the recognition subroutine of \cite{CSG05}'s \texttt{LAYOUT\_LEFT\_TO\_RIGHT} algorithm.)

We chose to use relatively small $n \in \{12, 14, 16, 18\}$ for comparison against \cite{GS84}---although their algorithm works for much larger $n$ when $k$ is small, its $O(n^k)$ time and space complexity render it ineffective when $k$ is large, restricting us to this range. Meanwhile, we used larger $n \in \{40, 45, 50, 55\}$ for comparison against \cite{DCM99} and \cite{CSG05}. For every fixed $n$, we varied over several large $k$ values and generated five affirmative and five negative test cases in this manner for each $(n, k)$ pair. In particular, we used $k \in \{n - 6, n - 4, n - 2\}$ for the affirmative cases (resulting in $15$ affirmative cases per $n$) and $k \in \{n - 6, n - 4\}$ for the negative cases (resulting in $10$ negative cases per $n$). (It would have been incredibly computationally expensive to generate a graph with bandwidth $> n - 3$ using the methodology described above without any bias toward certain structures, so we limited $k$ to at most $n - 4$ in the negative test cases.)

For each instance, we first ran each algorithm in a separate Julia subprocess with a $90$-second timeout to confirm that the affirmative/negative verdict matched what was expected, then benchmarked more accurately in the parent process using the Julia package BenchmarkTools.jl (v1.6.3) \cite{CR16}. (We used subprocesses rather than asynchronous execution because it is difficult and often unreliable to impose hard time limits without subprocess isolation, especially when the computations in question are CPU-intensive.) Higher timeout limits (e.g., $120$ seconds) led to out-of-memory errors and system freezes when running \cite{GS84}'s algorithm due to its $O(n^k)$ space requirement; although this was not an issue for the other algorithms, we decided to enforce this $90$-second limit universally for consistency. Furthermore, as supported by both \cite[Table 5]{CSG05} and our own benchmark data below in Table \ref{table:3}, the MB-ID algorithm from \cite{DCM99} is either very fast or very slow, being far less consistent than the other algorithms. Thus, our $90$-second limit prevented runaway executions of MB-ID.

We used BenchmarkTools.jl to automatically determine the optimal number of executions per timing measurement for each benchmark, collect multiple such measurements, and take the minimum runtime. (As justified in \cite{CR16}, the minimum estimator given by BenchmarkTools.jl provides a more reliable estimate of true execution time than the mean or median by minimizing contributions from environmental delays such as operating system jitter, cache effects, and context switches.) We observed that the subprocess timeout limit of $90$ seconds translated to an effective benchmark time cap of approximately $60$--$75$ seconds in the parent process, accounting for both the overhead of running a subprocess and BenchmarkTools.jl's optimizations (e.g., variable interpolation and reduced measurement overhead). (That is---if an algorithm run would have resulted in a benchmark time exceeding $60$--$75$ seconds in the parent process, then it would have taken over $90$ seconds in the initial subprocess due to the extra overhead and thus timed out before ever being benchmarked.)

Finally, after completing all algorithm runs, we saved the data in CSV files (available, alongside all companion code, at \cite{Var26}). Below, we report the number of instances solved, the number of instances exceeding the time limit, and (for solved instances) the mean and standard deviation of the minimum runtimes for each $(n, k)$ pair and algorithm tested.

\subsection{Results}\label{subsection:5.2}

Tables $1$--$4$ summarize our benchmark results. Each row corresponds to a fixed $(n, k)$ pair, and columns are grouped by algorithm. Within each group, the ``Solved'' and ``TLE'' (for time limit exceeded) columns indicate the number of instances that were and were not solved, respectively, while the ``Time'' column reports the mean and standard deviation of the minimum runtimes in milliseconds (computed only over solved instances). The bold entry in each row indicates the fastest mean runtime (regardless of the number of instances solved) across all algorithms for the corresponding $(n, k)$ pair.

\begin{table}[!ht]
    \centering
    \small
    \setlength{\tabcolsep}{4pt}
    \begin{tabular}{c|ccc|ccc}
        \toprule
        $(n, k)$ & \multicolumn{3}{c}{Our algorithm} & \multicolumn{3}{c}{$O(n^k)$ algorithm} \\
         & Solved & TLE & Time (ms) & Solved & TLE & Time (ms) \\
        \midrule
        $(12, 6)$ & $5$ & $0$ & {\boldmath$0.016 \pm 0.020$} & $5$ & $0$ & $13,\!762.4 \pm 6,\!463.82$ \\
        $(12, 8)$ & $5$ & $0$ & {\boldmath$0.0050 \pm 0.0001$} & $0$ & $5$ & $-$ \\
        $(12, 10)$ & $5$ & $0$ & {\boldmath$0.0048 \pm 0.0001$} & $0$ & $5$ & $-$ \\
        \midrule
        $(14, 8)$ & $5$ & $0$ & {\boldmath$2.19 \pm 3.80$} & $1$ & $4$ & $14,\!320.0$ \\
        $(14, 10)$ & $5$ & $0$ & {\boldmath$0.0069 \pm 0.0001$} & $0$ & $5$ & $-$ \\
        $(14, 12)$ & $5$ & $0$ & {\boldmath$0.0065 \pm 0.0001$} & $0$ & $5$ & $-$ \\
        \midrule
        $(16, 10)$ & $5$ & $0$ & {\boldmath$0.011 \pm 0.0051$} & $0$ & $5$ & $-$ \\
        $(16, 12)$ & $5$ & $0$ & {\boldmath$0.0089 \pm 0.0001$} & $0$ & $5$ & $-$ \\
        $(16, 14)$ & $5$ & $0$ & {\boldmath$0.0084 \pm 0.0002$} & $0$ & $5$ & $-$ \\
        \midrule
        $(18, 12)$ & $5$ & $0$ & {\boldmath$0.764 \pm 1.54$} & $0$ & $5$ & $-$ \\
        $(18, 14)$ & $5$ & $0$ & {\boldmath$0.011 \pm 0.0002$} & $0$ & $5$ & $-$ \\
        $(18, 16)$ & $5$ & $0$ & {\boldmath$0.011 \pm 0.0002$} & $0$ & $5$ & $-$ \\
        \bottomrule
    \end{tabular}
    \caption{Our algorithm vs. $O(n^k)$ algorithm \cite{GS84} (affirmative cases)}
    \label{table:1}
\end{table}

We begin with our benchmarks against \cite{GS84}'s $O(n^k)$ algorithm. Table \ref{table:1} presents results for affirmative test cases, where bandwidth is at most $k$. Our algorithm solved all $60$ instances across $(n, k)$ pairs, with runtimes consistently in the microsecond range for $k \in \{n - 4, n - 2\}$ and still in the millisecond range at worst for $k = n - 6$. In contrast, the $O(n^k)$ algorithm timed out (exceeding our $90$-second limit) on $54$ of these instances, successfully completing only $6$ instances at the smallest values of $n$ and $k$ tested. Moreover, the instances solved by the $O(n^k)$ algorithm took over $10$ seconds on average, ranging from $5,\!000$--$1,\!000,\!000$ times slower than our algorithm on the same instances.

\begin{table}[!ht]
    \centering
    \small
    \setlength{\tabcolsep}{4pt}
    \begin{tabular}{c|ccc|ccc}
        \toprule
        $(n, k)$ & \multicolumn{3}{c}{Our algorithm} & \multicolumn{3}{c}{$O(n^k)$ algorithm} \\
         & Solved & TLE & Time (ms) & Solved & TLE & Time (ms) \\
        \midrule
        $(12, 6)$ & $5$ & $0$ & $53.9 \pm 0.756$ & $5$ & $0$ & {\boldmath$0.063 \pm 0.050$} \\
        $(12, 8)$ & $5$ & $0$ & {\boldmath$0.730 \pm 0.0097$} & $5$ & $0$ & $27.3 \pm 32.1$ \\
        \midrule
        $(14, 8)$ & $5$ & $0$ & $181 \pm 2.12$ & $5$ & $0$ & {\boldmath$0.322 \pm 0.528$} \\
        $(14, 10)$ & $5$ & $0$ & {\boldmath$1.63 \pm 0.055$} & $4$ & $1$ & $12,\!663.3 \pm 6,\!300.23$ \\
        \midrule
        $(16, 10)$ & $5$ & $0$ & $522 \pm 3.08$ & $5$ & $0$ & {\boldmath$11.8 \pm 25.1$} \\
        $(16, 12)$ & $5$ & $0$ & {\boldmath$2.85 \pm 0.090$} & $0$ & $5$ & $-$ \\
        \midrule
        $(18, 12)$ & $5$ & $0$ & {\boldmath$1,\!145.39 \pm 19.3$} & $4$ & $1$ & $2,\!781.58 \pm 5,\!311.12$ \\
        $(18, 14)$ & $5$ & $0$ & {\boldmath$4.40 \pm 0.042$} & $0$ & $5$ & $-$ \\
        \bottomrule
    \end{tabular}
    \caption{Our algorithm vs. $O(n^k)$ algorithm \cite{GS84} (negative cases)}
    \label{table:2}
\end{table}

Table \ref{table:2} presents corresponding results for negative test cases, where bandwidth is strictly greater than $k$. The comparison here is more nuanced. The $O(n^k)$ algorithm was substantially faster for $k = n - 6$ and $n \in \{12, 14\}$, consistently completing in under a millisecond whereas our algorithm took tens to hundreds of milliseconds. For $k = n - 6$ and $n = 16$, the $O(n^k)$ algorithm still outperformed ours (although with a less pronounced gap), but for $k = n - 6$ and $n = 18$, it solved only $4$ out of $5$ instances and actually performed worse even on solved instances. Meanwhile, the $O(n^k)$ algorithm solved some instances for $k = n - 4$ and $n \in \{12, 14\}$, but our algorithm consistently outperformed it by a wide margin. For all other $n$ and $k$, the $O(n^k)$ algorithm always exceeded the time limit while our algorithm always succeeded (indeed, in less than $1.2$ seconds).

Two observations here are worth noting. First, our algorithm consistently performs better on affirmative cases than on negative cases whereas the opposite holds for the $O(n^k)$ algorithm, and second, our algorithm exhibits much lower variance on negative cases. Both phenomena follow from the same underlying reason: our algorithm can terminate early on affirmative cases (upon finding a valid layout) but must spend $\Theta(n^2)$ time checking Theorem \ref{theorem:3.2}'s Hall conditions for all
\[\frac{n!}{(n - (n - k - 1))!} \in \Theta(n^{n - k - 1})\]
left partial layouts on negative cases, yielding consistent runtimes. On the other hand, the $O(n^k)$ algorithm can terminate early due to its pruning strategy (see \cite{GS84}), but the effectiveness of this approach varies wildly with graph structure. Therefore, our algorithm requires $\Theta(n^{n - k + 1})$ time when solving negative cases, whereas a tight lower bound on our competitor's time requirement is $\Omega(nk^2)$ when pruning succeeds at every step, thus requiring only $O(n)$ calls to the $O(k^2)$ ``Update'' procedure \cite[Figure 2.2]{GS84}.

\begin{table}[!ht]
    \centering
    \small
    \setlength{\tabcolsep}{4pt}
    \begin{tabular}{c|ccc|ccc|ccc}
        \toprule
        $(n, k)$ & \multicolumn{3}{c}{Our algorithm} & \multicolumn{3}{c}{MB-ID} & \multicolumn{3}{c}{\texttt{LAYOUT\_LEFT\_TO\_RIGHT}} \\
         & Solved & TLE & Time (ms) & Solved & TLE & Time (ms) & Solved & TLE & Time (ms) \\
        \midrule
        $(40, 34)$ & $5$ & $0$ & $1.48 \pm 1.88$ & $4$ & $1$ & {\boldmath$0.147 \pm 0.023$} & $5$ & $0$ & $1.36 \pm 0.650$ \\
        $(40, 36)$ & $5$ & $0$ & {\boldmath$0.106 \pm 0.0037$} & $5$ & $0$ & $31.6 \pm 70.3$ & $5$ & $0$ & $0.979 \pm 0.060$ \\
        $(40, 38)$ & $5$ & $0$ & {\boldmath$0.102 \pm 0.0042$} & $5$ & $0$ & $0.130 \pm 0.0020$ & $5$ & $0$ & $0.980 \pm 0.075$ \\
        \midrule
        $(45, 39)$ & $5$ & $0$ & {\boldmath$0.147 \pm 0.0045$} & $5$ & $0$ & $0.194 \pm 0.010$ & $5$ & $0$ & $1.54 \pm 0.170$ \\
        $(45, 41)$ & $5$ & $0$ & {\boldmath$0.144 \pm 0.0047$} & $5$ & $0$ & $0.187 \pm 0.0077$ & $5$ & $0$ & $1.69 \pm 0.319$ \\
        $(45, 43)$ & $5$ & $0$ & {\boldmath$0.144 \pm 0.0065$} & $5$ & $0$ & $0.179 \pm 0.0043$ & $5$ & $0$ & $1.46 \pm 0.131$ \\
        \midrule
        $(50, 44)$ & $5$ & $0$ & {\boldmath$0.203 \pm 0.0068$} & $5$ & $0$ & $5,\!075.13 \pm 11,\!347.7$ & $5$ & $0$ & $2.46 \pm 0.588$ \\
        $(50, 46)$ & $5$ & $0$ & {\boldmath$0.206 \pm 0.0062$} & $4$ & $1$ & $0.248 \pm 0.0020$ & $5$ & $0$ & $1.88 \pm 0.185$ \\
        $(50, 48)$ & $5$ & $0$ & {\boldmath$0.196 \pm 0.0054$} & $5$ & $0$ & $0.239 \pm 0.0041$ & $5$ & $0$ & $1.95 \pm 0.130$ \\
        \midrule
        $(55, 49)$ & $5$ & $0$ & {\boldmath$0.341 \pm 0.158$} & $3$ & $2$ & $0.351 \pm 0.044$ & $5$ & $0$ & $3.49 \pm 1.06$ \\
        $(55, 51)$ & $5$ & $0$ & {\boldmath$0.268 \pm 0.0074$} & $4$ & $1$ & $0.322 \pm 0.0076$ & $5$ & $0$ & $2.59 \pm 0.173$ \\
        $(55, 53)$ & $5$ & $0$ & {\boldmath$0.269 \pm 0.0089$} & $5$ & $0$ & $0.320 \pm 0.0058$ & $5$ & $0$ & $2.45 \pm 0.169$ \\
        \bottomrule
    \end{tabular}
    \caption{Our algorithm vs. MB-ID \cite{DCM99} vs. \texttt{LAYOUT\_LEFT\_TO\_RIGHT} \cite{CSG05} (affirmative cases)}
    \label{table:3}
\end{table}

Table \ref{table:3} compares our algorithm against the recognition subroutines of the MB-ID \cite{DCM99} and \texttt{LAYOUT\_LEFT\_TO\_RIGHT} \cite{CSG05} bandwidth minimization algorithms on affirmative cases, this time with much larger $n$. Our algorithm solved all $60$ instances and achieved the fastest mean runtime in $11$ of the $12$ $(n, k)$ pairs (only losing to MB-ID for $n = 40$ and $k = n - 6$). \texttt{LAYOUT\_LEFT\_TO\_RIGHT} also solved all instances but was consistently slower than our algorithm, typically by a factor of $10$. MB-ID timed out on $5$ of the $60$ instances and exhibited highly variable performance, ranging from the sub-millisecond range to tens of seconds even for the same $(n, k)$ pairs. When MB-ID did perform well, it was typically slightly slower than our algorithm but still faster than \texttt{LAYOUT\_LEFT\_TO\_RIGHT}. (The raw data available at \cite{Var26}, with its greater degree of granularity, shows these patterns more clearly.)

\begin{table}[!ht]
    \centering
    \small
    \setlength{\tabcolsep}{4pt}
    \begin{tabular}{c|ccc|ccc|ccc}
        \toprule
        $(n, k)$ & \multicolumn{3}{c}{Our algorithm} & \multicolumn{3}{c}{MB-ID} & \multicolumn{3}{c}{\texttt{LAYOUT\_LEFT\_TO\_RIGHT}} \\
         & Solved & TLE & Time (ms) & Solved & TLE & Time (ms) & Solved & TLE & Time (ms) \\
        \midrule
        $(40, 34)$ & $0$ & $5$ & $-$ & $5$ & $0$ & {\boldmath$0.769 \pm 0.451$} & $5$ & $0$ & $21.4 \pm 14.7$ \\
        $(40, 36)$ & $5$ & $0$ & $127 \pm 2.91$ & $5$ & $0$ & {\boldmath$0.512 \pm 0.135$} & $5$ & $0$ & $11.8 \pm 4.49$ \\
        \midrule
        $(45, 39)$ & $0$ & $5$ & $-$ & $5$ & $0$ & {\boldmath$1.28 \pm 1.06$} & $5$ & $0$ & $38.9 \pm 37.0$ \\
        $(45, 41)$ & $5$ & $0$ & $220 \pm 4.93$ & $5$ & $0$ & {\boldmath$0.899 \pm 0.171$} & $5$ & $0$ & $25.6 \pm 4.87$ \\
        \midrule
        $(50, 44)$ & $0$ & $5$ & $-$ & $5$ & $0$ & {\boldmath$1.33 \pm 0.500$} & $5$ & $0$ & $42.9 \pm 17.7$ \\
        $(50, 46)$ & $5$ & $0$ & $306 \pm 3.76$ & $5$ & $0$ & {\boldmath$1.25 \pm 0.176$} & $5$ & $0$ & $43.4 \pm 7.56$ \\
        \midrule
        $(55, 49)$ & $0$ & $5$ & $-$ & $5$ & $0$ & {\boldmath$2.04 \pm 0.716$} & $5$ & $0$ & $76.3 \pm 33.2$ \\
        $(55, 51)$ & $5$ & $0$ & $469 \pm 4.38$ & $5$ & $0$ & {\boldmath$1.86 \pm 0.325$} & $5$ & $0$ & $69.0 \pm 14.2$ \\
        \bottomrule
    \end{tabular}
    \caption{Our algorithm vs. MB-ID \cite{DCM99} vs. \texttt{LAYOUT\_LEFT\_TO\_RIGHT} \cite{CSG05} (negative cases)}
    \label{table:4}
\end{table}

Finally, Table \ref{table:4} presents corresponding results for negative cases. Here, our algorithm's $\Theta(n^{n - k + 1})$ time complexity for negative cases---never breaking early---becomes a significant limitation. For $k = n - 6$, this amounts to $\Theta(n^7)$ time, causing timeouts on all $20$ instances due to the moderately large $n$ values. Even for $k = n - 4$, where our algorithm successfully solved all $20$ instances, it consistently required hundreds of milliseconds---typically about $250$ times slower than MB-ID and $10$ times slower than \texttt{LAYOUT\_LEFT\_TO\_RIGHT}.

In summary, our algorithm demonstrates clear advantages for affirmative cases across all values of $n$ and $k$ tested, consistently outperforming or at least matching the fastest existing approaches. Therefore, when it is hoped or expected that a graph on $n$ nodes has bandwidth at most $k$ and $k$ is close to $n$, our algorithm is the best available choice. For negative cases, however, existing approaches (particularly MB-ID) are preferable even when $n - k$ is relatively small, since our algorithm must exhaustively check all $\Theta(n^{n - k - 1})$ left partial layouts without the same opportunity for early termination that other algorithms' pruning strategies allow.

\section{Conclusion and future work}\label{section:6}

Motivated by applications of bandwidth recognition algorithms to fields like computational complexity theory \cite{RS83} and quantum information theory \cite{JMP25}---as well as the problem's general theoretical significance---we have developed the first algorithm specifically designed for high-bandwidth recognition. Given a graph $G$ on $n$ nodes and an integer $k \ge \left\lfloor \frac{n - 1}{2} \right\rfloor$, our algorithm leverages Hall's marriage theorem to determine whether $\beta(G) \le k$ in $O(n^{n - k + 1})$ time and $O(n)$ auxiliary space (beyond storage of the input graph). Just as \cite{GS84}'s classical $O(n^k)$ time and space algorithm is polynomial in complexity when $k$ is small and close to $0$, our algorithm is polynomial when $k$ is large and close to the maximum possible bandwidth of $n - 1$.

Our benchmark results show that our algorithm outperforms all other existing approaches when $n - k$ is small (namely, between $4$ and $6$) and $\beta(G)$ is indeed at most $k$. (Admittedly, our algorithm does suffer somewhat when $\beta(G) > k$ due to the lack of pruning strategies compared to others.) The most performant competitors were the recognition subroutines in the minimization algorithms presented by \cite{DCM99} and \cite{CSG05}, both of which were consistently slower than our algorithm. In particular, our algorithm achieved the fastest mean runtimes in $11$ of the $12$ test suites, each comprising $5$ random graphs on $n$ nodes with bandwidth at most $k$ (for a particular choice of $n$ and $k$).

These results point to several directions for future work. Integrating pruning strategies (see \cite{GS84,DCM99,CSG05}) to detect infeasibility earlier during left partial layout enumeration could improve performance on negative cases. Extending our approach to weighted graphs is another avenue---it might be possible to construct an analogous bandwidth recognition algorithm when bandwidth is defined as the maximum product of edge weight and label difference across all edges. Finally, it is worth exploring whether Hall's marriage theorem can be applied to the problem of low-bandwidth recognition as well.

\section*{Acknowledgements}

I would like to thank Nathaniel Johnston, Liam Keliher, and Laurie Ricker for valuable advice and guidance. Additional gratitude is owed to Alexander B. Matheson, Serena Davis, Luc Campbell, Marco Cognetta, and Ryan Tifenbach for further helpful conversations.

\bibliographystyle{alpha}
\bibliography{refs}

\end{document}